\documentclass[a4paper,DIV=10,11pt,titlepage=off]{scrartcl}

\usepackage{graphicx}
\usepackage{mathpazo}
\usepackage[margin=1in]{geometry}
\usepackage{amsmath,amssymb,amsthm}
\usepackage{booktabs}
\usepackage{enumitem}
\usepackage[expansion=false]{microtype}
\usepackage[dvipsnames]{xcolor}
\usepackage[colorlinks=true,linkcolor=blue!55!black,urlcolor=blue!55!black,
            citecolor=blue!55!black]{hyperref}
\usepackage[capitalise,noabbrev]{cleveref}
\usepackage{authblk}

\theoremstyle{plain}
\newtheorem{theorem}{Theorem}[section]
\newtheorem{lemma}[theorem]{Lemma}
\newtheorem{proposition}[theorem]{Proposition}
\newtheorem{corollary}[theorem]{Corollary}
\newtheorem{conjecture}[theorem]{Conjecture}

\theoremstyle{definition}

\newtheorem{problem}[theorem]{Problem}
\newtheorem{example}[theorem]{Example}
\theoremstyle{remark}
\newtheorem{remark}[theorem]{Remark}

\crefname{problem}{Problem}{Problems}
\Crefname{problem}{Problem}{Problems}
\crefname{assumption}{Assumption}{Assumptions}
\Crefname{assumption}{Assumption}{Assumptions}
\crefname{conjecture}{Conjecture}{Conjectures}
\Crefname{conjecture}{Conjecture}{Conjectures}

\providecommand{\ket}[1]{\lvert #1 \rangle}

\DeclareMathOperator{\Cl}{Cl}
\DeclareMathOperator{\sd}{sd}
\DeclareMathOperator{\rank}{rank}
\DeclareMathOperator{\spec}{spec}
\DeclareMathOperator{\Tor}{Tor}
\DeclareMathOperator{\Ind}{Ind}
\newcommand{\poly}{\operatorname{poly}}
\newcommand{\FF}{\mathbb F}
\newcommand{\ZZ}{\mathbb Z}
\newcommand{\QQ}{\mathbb Q}
\newcommand{\RR}{\mathbb R}

\newcommand{\Ht}{\widetilde H}
\newcommand{\Ct}{\widetilde C}
\newcommand{\bt}{\widetilde\beta}

\newcommand{\QMA}{\mathsf{QMA}}
\newcommand{\QMAone}{\mathsf{QMA}_1}
\newcommand{\DQCone}{\mathsf{DQC1}}
\newcommand{\NP}{\mathsf{NP}}
\newcommand{\DCH}{\textup{\textsc{CliqueHom}}}
\newcommand{\DCHtwo}{\textup{\textsc{CliqueHom}}_{\FF_2}}
\newcommand{\DCHtfno}{\textup{\textsc{CliqueHom}}^{\mathrm{tf}\text{-}\textsc{no}}}
\newcommand{\GTD}{\textup{\textsc{GappedTorsion}}}

\newcommand{\yes}{\textup{\textsc{yes}}}
\newcommand{\no}{\textup{\textsc{no}}}

\title{Torsion detection in clique complexes is conditionally $\QMAone$-hard}
\author{Adam Wesołowski\thanks{\texttt{Adam.Wesolowski@cs.ox.ac.uk}}}
\affil{\small \textit{Department of Computer Science, University of Oxford, Parks Rd, Oxford OX1 3QG, United Kingdom}}
\date{}

\begin{document}
\maketitle
\vspace{-1.4cm}
\begin{abstract}
Quantum algorithms for topological data analysis compute Betti numbers, the ranks of the
homology groups of a simplicial complex, which can be read off from the kernel of a
combinatorial Laplacian. Deciding whether a Betti number of a clique complex is nonzero
is $\QMAone$-hard, and remains so under a spectral gap promise on vertex-weighted graphs
\cite{CK24,KK24}. Integral homology, however, contains information inaccessible to the Laplacian
spectrum. A new part that appears in integral homology is \emph{torsion}: cycles that become boundaries only after being
traversed several times, as in a projective plane or a Klein bottle.

We provide a simple reduction that allows us to establish the first hardness results for the problem of detecting torsion. We attach
to an arbitrary clique complex a fixed $31$-vertex triangulation of the projective plane.
The $k$th mod-$2$ Betti number of the input then reappears as $2$-torsion two degrees up,
while all rational homology disappears and every combinatorial Laplacian acquires a
constant spectral gap. We conclude that detecting torsion in clique complexes of
\emph{unweighted} graphs is $\NP$-hard, even under a constant gap promise, and that it is
$\QMAone$-hard if mod-$2$ clique homology is. 

\end{abstract}

\tableofcontents

\section{Introduction}
\label{sec:intro}

Topological data analysis (TDA) asks how many connected
components, loops, voids and higher-dimensional holes does the space contain?  Over the integers,
the $k$th homology group splits as
\[
H_k(X;\mathbb{Z}) \;\cong\; \mathbb{Z}^{\beta_k} \oplus \operatorname{Tor} H_k(X;\mathbb{Z}),
\]
where $\beta_k$ is the $k$th Betti number and $\operatorname{Tor} H_k(X;\mathbb{Z})$ is a
finite torsion subgroup. Passing to real coefficients kills the torsion summand outright,
$H_k(X;\mathbb{R})\cong\mathbb{R}^{\beta_k}$, so two complexes with identical Betti numbers
can still differ once integral coefficients are restored.

Quantum TDA started with Lloyd, Garnerone and Zanardi \cite{LGZ16} and was refined over
the following decade \cite{GK19,GCD22,Hay22,UAS+21,MGB26,BSG+24}. It uses a
linear-algebraic route to Betti numbers. The number $\beta_k$ is the dimension of the
kernel of the combinatorial Laplacian $\Delta_k=B_k^{*}B_k+B_{k+1}B_{k+1}^{*}$, built from
the real boundary maps $B_k$ of the complex; since the kernel is computed over
$\mathbb{R}$, it recovers only the free rank $\beta_k$ and is structurally blind to
torsion. The Laplacian is sparse and can be block encoded from an adjacency oracle.
Counting holes therefore becomes a ground-space problem for a Hermitian operator on
potentially exponentially many simplices. Crichigno and Kohler \cite{CK24} showed that
deciding whether a clique complex has a $k$-dimensional hole is $\QMAone$-hard. King and
Kohler \cite{KK24} showed the same for the \emph{gapped} version on vertex-weighted
graphs, and placed that version in $\QMA$.

\paragraph{Integral torsion.}
Integral homology records qualitatively different information than the number of holes in
a topological space. Some cycles are not boundaries, and never become boundaries however
many times you go around them. The Betti number counts these cycles. Other cycles are not
boundaries, but going around them $d$ times \emph{is} a boundary. These are torsion
classes, and $d$ is the torsion order. The model example is the loop $\ell$ of the
projective plane. It is not a boundary, but traversing it twice sweeps out the whole
surface, so $2\ell$ is. Hence $\Ht_1(\mathbb{RP}^2;\ZZ)\cong\ZZ/2$. Torsion records how a
space is glued to itself, or in other words how it is `twisted'.

Over the real numbers torsion disappears. From $dz=\partial w$ one divides by $d$ and
concludes that $z$ is a boundary. Hodge theory needs an inner product, an inner product
needs $\RR$, and $\RR$ cannot see torsion. Reducing the Laplacian modulo a prime does not
help either. Over $\FF_p$ the bilinear form on chains is degenerate, so the Hodge
decomposition fails and $\ker(\Delta_k\bmod p)$ is not the mod-$p$ homology. A single edge
already has $\Delta_1\equiv0\bmod2$ while $H_1=0$. So every Laplacian-based algorithm,
classical or quantum, is blind to torsion. It cannot tell a Klein bottle from a circle.
Both have a single rational $1$-cycle, and the Klein bottle differs only by a class of
order two.

We will repeatedly need the opposite situation. We say a complex
is \emph{$2$-torsion-free in degree $k$} if its integral homology group $\Ht_k(X;\ZZ)$
contains no element of order $2$. In other words, there is no $k$-cycle that fails to
bound but whose double bounds. Most familiar spaces are $2$-torsion-free in every degree,
for example spheres, tori, orientable surfaces and wedges of spheres. The projective plane
and the Klein bottle are the basic examples that are not. The point of the notion is the
universal coefficient theorem. If $X$ is $2$-torsion-free in degrees $k-1$ and $k$, then
its mod-$2$ and rational Betti numbers in degree $k$ coincide. Counting holes over $\FF_2$
and over $\QQ$ then gives the same answer. When $2$-torsion is present the mod-$2$ count
is strictly larger.

Torsion is not an exotic invariant. Carlsson, Ishkhanov, de Silva and Zomorodian found
that the space of high-contrast $3\times3$ patches of natural images concentrates on a
Klein bottle \cite{CIdSZ08}. The conformation space of cyclo-octane, a benchmark in
computational chemistry, is a sphere and a Klein bottle glued along two circles
\cite{MTCW10}. The Klein bottle component is certified in data by computing persistent
Betti numbers over $\FF_2$ and $\FF_3$ and seeing them disagree \cite{MPSBF19}. Models of
orientation preference in the visual cortex map the stimulus space onto a Klein bottle
\cite{Swi96}. In each case the interesting signal is torsion. Classical
persistent-homology software therefore offers coefficients in arbitrary prime fields
\cite{Bau21}. The effect of the choice of field has been analysed in detail \cite{OY23},
and integral persistence has been developed \cite{RRS14}. Torsion of simplicial complexes
is also a classical object in its own right \cite{Lor08,DKM09,DKM13,New19}. Classically,
on an explicitly given complex, torsion is computed in polynomial time by Smith normal
form.

\paragraph{A new candidate for quantum advantage.}
To our knowledge, the computational complexity of evaluating or even detecting torsion
has never been studied. We are not aware of any hardness results or of any prior work on quantum
algorithms. Whether torsion is a promising target for quantum advantage is therefore an
open question, and it remains entirely plausible that quantum algorithms for this
property of topological spaces might exhibit quantum advantage. In this work we present
the first quantum complexity analysis of the torsion detection problem.

\paragraph{Our contributions.}
We give a reduction that translates the mod-$2$ Betti number of the input into nontrivial
torsion of the output. The gadget we use is the six-vertex projective plane, the standard
example of a rationally acyclic complex with torsion in the simplicial matrix-tree
literature \cite{Kal83,DKM13}, made into a clique complex by the so-called
\emph{barycentric subdivision}.
This work builds and extends on several key innovations and observations of a preliminary version of \cite{TLW+26}. While \cite{TLW+26} concentrates on building a practical algorithm, this work takes a different direction and concentrates on complexity aspect.

The homology of a join complex is given by the K\"unneth
formula. The Laplacian of a join is a Kronecker sum, which is how \cite{KK24,BSG+24} build
gapped benchmark families. We show that joining the gadget onto an \emph{arbitrary}
clique complex plants its mod-$2$ Betti number as $2$-torsion two degrees up, and at the
same time yields a constant spectral gap on every Laplacian of the join. That is enough
to transfer hardness from mod-$2$ Betti number estimation to torsion detection.
\begin{itemize}[leftmargin=*,itemsep=2pt]
\item \emph{A gadget that turns holes into torsion.} We join the input complex $X$ with a
fixed $31$-vertex clique complex $Y$ triangulating the projective plane. We prove that
the join has no rational homology at all, that its integral homology in degree $k+2$ is
$(\ZZ/2)^{\bt_k(X;\FF_2)}$, and that every one of its Laplacians has smallest eigenvalue
at least an absolute constant $c_Y\approx0.0979$. Our reduction adds $31$ vertices and
answers each adjacency query with one adjacency query to $X$.
\item \emph{Hardness.} We prove that detecting torsion in the clique complex of an
unweighted graph, under a constant Laplacian-gap promise, is $\NP$-hard, and that
\textbf{it is at least as hard as mod-$2$ clique homology}. We prove that it is
$\QMAone$-hard under a natural conjecture, namely that clique homology remains
$\QMAone$-hard when its \no\ instances are promised to have no $2$-torsion in the two
relevant degrees (\cref{thm:main,cor:np-hard}).
\end{itemize}
\Cref{fig:overview} summarises our reduction.

\begin{figure}[t]
\centering
\includegraphics[width=\linewidth]{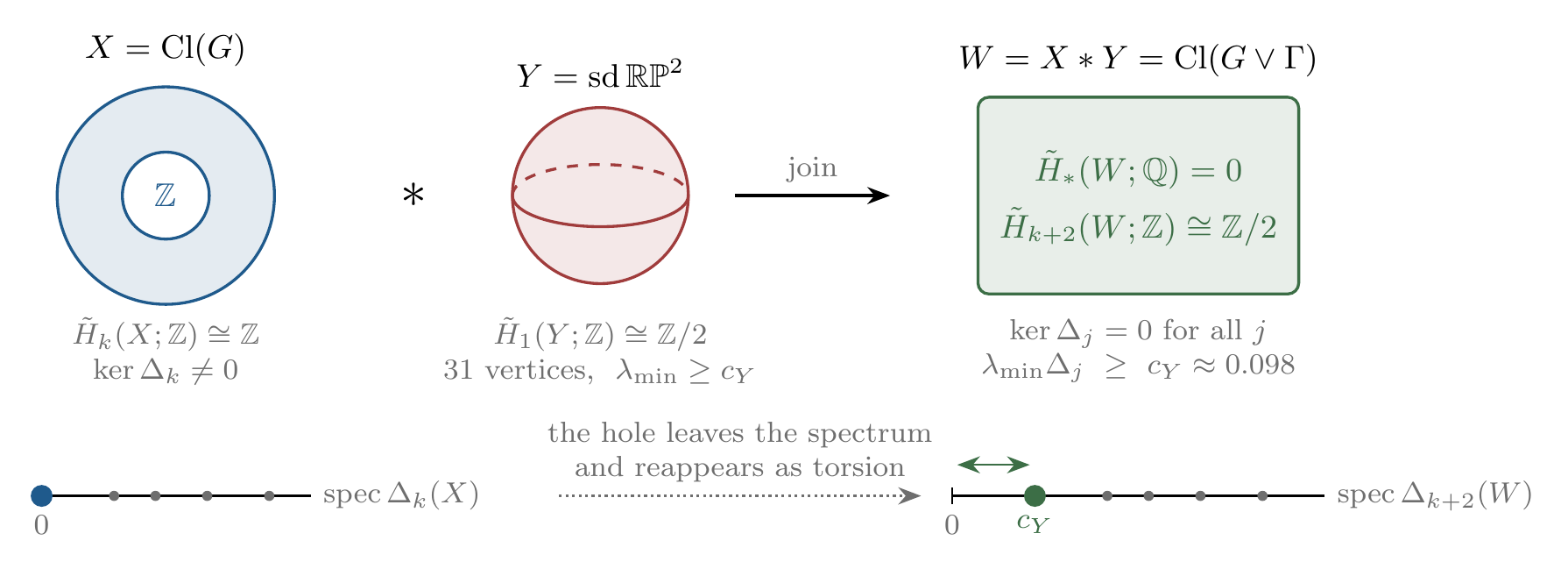}
\caption{Overview. Joining the input with a fixed projective plane removes all rational
homology, gaps every Laplacian, and re-encodes the $k$th mod-$2$ Betti number as
$2$-torsion in degree $k+2$. }
\label{fig:overview}
\end{figure}

\subsection{Results}
\label{sec:results}

Throughout, graphs are given by an adjacency oracle, the efficient access model of
\cite{CK24,KK24}. A graph on $N$ vertices specifies a complex with up to
$\binom{N}{k+1}$ simplices in degree $k$. We prove the following.

\begin{theorem}[Hardness, informal version of \cref{thm:main,cor:np-hard}]
\label{thm:main-informal}
Let\\ $\GTD_{n,g}$ be the following problem. Given an unweighted graph $H$ and an integer
$n$, under the promise that $\Ht_n(\Cl(H);\ZZ)$ is finite and that every augmented
Laplacian of $\Cl(H)$ has smallest eigenvalue at least $g$, decide whether
$\Ht_n(\Cl(H);\ZZ)\neq0$. Then, with the absolute constant $g=c_Y\approx0.0979$:
\begin{enumerate}[label=(\roman*),leftmargin=*,itemsep=1pt]
\item $\GTD_{n,c_Y}$ is $\NP$-hard.
\item $\GTD_{k+2,c_Y}$ is at least as hard as deciding whether
$\Ht_k(\Cl(G);\FF_2)\neq0$ for unweighted $G$.
\item Under \cref{conj:torsionfree}, $\GTD_{n,c_Y}$ is $\QMAone$-hard.
\end{enumerate}
Our reduction adds $31$ vertices. It answers one adjacency query to its output with one
adjacency query to its input and $O(1)$ reversible gates.
\end{theorem}

\subsection{Technical overview}
\label{sec:overview}

\paragraph{The question.}
A quantum TDA algorithm looks at a simplicial complex through the real spectrum of its
Laplacian matrices. Torsion does not show up in that window. We show how to \emph{hide} a
known hard problem about Betti numbers inside a problem about torsion.

\paragraph{A gadget that turns holes into torsion.}
We take any complex $X$ and a fixed small complex $Y$ that is a triangulated projective
plane, and we join $X$ with $Y$, that is, we connect every vertex of one to every vertex
of the other. The homology of the new complex is computed from that of the factors by a
K\"unneth formula. The projective plane has no rational homology at all, only a single
class of order two. The K\"unneth formula multiplies every mod-$2$ homology class of $X$
by that order-two class. The product is a class of order two in the join, two degrees up.
As a consequence the join has \emph{no} rational homology in any degree, and its integral
torsion in degree $n+2$ records exactly the $n$th mod-$2$ Betti number of $X$
(\cref{lem:planting}). A hole in $X$ has become torsion in $X*Y$.

\paragraph{The gadget also fixes the conditioning.}
The Laplacian of a join is not a complicated object. On each piece of the join it is the
sum of a Laplacian of $X$ and a Laplacian of $Y$, acting on different tensor factors. Its
eigenvalues are therefore sums of an eigenvalue of $X$ and an eigenvalue of $Y$
(\cref{lem:join-laplacian}). The eigenvalues of $X$ are nonnegative. The eigenvalues of
$Y$ are bounded below by a positive number $c_Y$, because $Y$ has no rational homology.
Hence \emph{every} Laplacian of the join has all its eigenvalues at least $c_Y$, whatever
$X$ is (\cref{prop:gap}). We obtain the constant spectral gap as a by-product of the same
operation that plants the torsion.

\paragraph{The reduction.}
We start from the clique-homology problem over the rationals of \cite{CK24}. Given a
graph, does its clique complex have a hole in degree $k$? We join with the gadget. If
there was a hole, the join has torsion. If there was no rational hole \emph{and} no
$2$-torsion in the two adjacent degrees, the join has no torsion. The output is a clique
complex of an unweighted graph. Every Laplacian of it has a constant gap, all its
rational Betti numbers are zero, and its torsion in one degree encodes the answer. A
spectral algorithm run on it returns ``no homology'' every time. The one caveat is the
$2$-torsion-freeness of the \no\ instances, and only of the \no\ instances. This is a
property of the hard instances of \cite{CK24} that we do not know how to verify, and we
isolate it as \cref{conj:torsionfree}. Without it, our reduction still shows that gapped
torsion detection is at least as hard as \emph{mod-$2$} clique homology. That problem is
$\NP$-hard on instances that are known to be torsion-free \cite{AS12,AS16}. So the
$\NP$-hardness we obtain is unconditional, and the $\QMAone$-hardness is conditional on a
single natural hypothesis.

\paragraph{When torsion is easy.}
Our gadget also delimits the hardness. On instances that are joins with the gadget whose
other factor is $2$-torsion-free in the two relevant degrees, the $2$-rank of the planted
torsion collapses to a rational Betti number of that factor (\cref{rem:easy}). On such
families, torsion carries exactly the information that Betti number estimation extracts.

\subsection{Related work}
\label{sec:related}

Crichigno and Kohler \cite{CK24} consider the plain decision problem. Given an unweighted
graph $G$ and an integer $k$, is $\Ht_k(\Cl(G);\QQ)\neq0$? They show it is
$\QMAone$-hard, and that a promise version of it lies in $\QMA$. King and Kohler
\cite{KK24} consider the \emph{gapped} problem, in which \yes\ instances have
$\ker\Delta_k\neq0$ and \no\ instances satisfy $\lambda_{\min}(\Delta_k)\ge g\ge1/\poly(n)$.
They prove it $\QMAone$-hard and contained in $\QMA$. Their reduction needs
\emph{vertex-weighted} graphs, and the unweighted case is left open. The gap promise is
not mild. Black, Maxwell and Nayyeri \cite[Theorems~1.2, 1.3 and \S6.4.1]{BMN23} exhibit
complexes, including clique-dense ones, whose Laplacian gap is exponentially small.
Schmidhuber and Lloyd \cite{SL23} and Berry et al.\ \cite{BSG+24} analyse when the
normalised Betti-number estimation problem solved by quantum TDA algorithms can offer an
advantage. Gyurik, Cade and Dunjko \cite{GCD22} show $\DQCone$-hardness of low-lying
spectral density estimation, and Cade and Crichigno \cite{CC24} show it for Betti-number
estimation of general chain complexes. The six-vertex projective plane and its role as a
$\mathbb Q$-acyclic complex with torsion come from the enumeration of simplicial spanning
trees \cite{Kal83,DKM09,DKM13}. Lowe, Kim, Bondesan and Hayakawa \cite{LKBH26} prove
$\DQCone$-hardness of low-energy spectral-density problems for clique-complex Laplacians,
the first such results directly on TDA instances. They leave approximate Betti-number
estimation itself open. Gyurik, Schmidhuber, King, Dunjko and Hayakawa \cite{GSKDH24} show
$\mathsf{BQP}_1$-hardness of a harmonic-persistence problem via the guided Hamiltonian.
None of these works addresses the complexity of torsion in the succinct access model.

\section{Preliminaries}
\label{sec:background}

\subsection{Complexes, boundary maps and Laplacians}
\label{sec:notation}

We write $[N]=\{1,\dots,N\}$ and $\binom{[N]}{j}$ for the $j$-subsets of $[N]$. All
complexes are finite abstract simplicial complexes. A face with $r+1$ vertices is an
$r$-simplex. For a graph $G=([N],E)$ the clique (flag) complex $\Cl(G)$ has a simplex for
every vertex set that induces a clique. We write $\overline G$ for the complement graph.

\paragraph{Augmented chains.}
We use the augmented (reduced) convention throughout. The empty face is the unique
$(-1)$-simplex and $\Ct_{-1}(X;\ZZ)\cong\ZZ$. The map $B_0$ sends every vertex to the
empty face, and $B_r$ is the usual signed boundary map for $r\ge1$. We write $\Ht_q(X;R)$
for reduced homology with coefficients in the ring $R$. For nonempty $X$ this differs from
unreduced homology only in degree $0$. Over $\RR$, with the simplex basis declared
orthonormal, the augmented Hodge Laplacians are
\[
\Delta_r=B_r^{*}B_r+B_{r+1}B_{r+1}^{*},\qquad -1\le r\le\dim X,
\]
with the convention $B_{-1}=0$ and $B_{\dim X+1}=0$. The Hodge theorem gives
$\ker\Delta_r\cong\Ht_r(X;\RR)$. With the unaugmented convention, $\Delta_0$ of a connected
complex would have a one-dimensional kernel. The augmented convention removes it, which is
what makes the gap statements below clean. We write $L^{\mathrm{up}}_{r}=B_{r+1}B_{r+1}^{*}$
and $L^{\mathrm{down}}_{r}=B_r^{*}B_r$, so $\Delta_r=L^{\mathrm{up}}_r+L^{\mathrm{down}}_r$.
Since $B_rB_{r+1}=0$, the two summands have orthogonal images and are simultaneously
diagonalisable. Hence every nonzero eigenvalue of $L^{\mathrm{up}}_{r}$ is an eigenvalue
of $\Delta_r$, and the positive squared singular values of $B_{r+1}$ are exactly the
nonzero eigenvalues of $L^{\mathrm{up}}_r$. We write $\zeta^{+}_{\min}(B)$ for the smallest
positive singular value of $B$.

\subsection{Betti numbers and torsion}
\label{sec:torsion}

For a field $F$ we let $\bt_k(X;F)=\dim_F\Ht_k(X;F)$. Real, rational and integral ranks
agree: $\bt_k(X;\RR)=\bt_k(X;\QQ)=\rank_\ZZ\Ht_k(X;\ZZ)$. We write $t_k(X)$ for the
$2$-rank of the torsion subgroup of $\Ht_k(X;\ZZ)$, that is, the number of its invariant
factors of even order. The universal coefficient theorem \cite[Theorem~3A.3]{Hat02} says
\[
\Ht_k(X;\FF_2)\;\cong\;\bigl(\Ht_k(X;\ZZ)\otimes\FF_2\bigr)\;\oplus\;
\Tor\bigl(\Ht_{k-1}(X;\ZZ),\FF_2\bigr).
\]
Together with $\ZZ/m\otimes\FF_2\cong\Tor(\ZZ/m,\FF_2)\cong\FF_2$ for even $m$ and $0$ for
odd $m$, this gives the count we use repeatedly:
\begin{equation}
\bt_k(X;\FF_2)\;=\;\bt_k(X;\QQ)\;+\;t_k(X)\;+\;t_{k-1}(X).
\label{eq:uct-count}
\end{equation}
In words, the mod-$2$ Betti number is the rational Betti number plus the number of
$2$-torsion classes in this degree and the one below. In particular
$\bt_k(X;\FF_2)\ge\bt_k(X;\QQ)$ always. Equality holds if and only if $\Ht_{k-1}(X;\ZZ)$
and $\Ht_k(X;\ZZ)$ contain no element of order $2$. This is why classical software detects
torsion by comparing Betti numbers over two fields.

\subsection{Joins and clique complexes}
\label{sec:joins}

The join $X*Y$ of complexes on disjoint vertex sets has faces $\sigma\cup\tau$ with
$\sigma\in X\cup\{\emptyset\}$ and $\tau\in Y\cup\{\emptyset\}$. The graph join
$G\vee\Gamma$ is the disjoint union of $G$ and $\Gamma$ together with every edge between
the two vertex sets. A graph is \emph{join-irreducible} if it is not isomorphic to
$G_1\vee G_2$ with both $G_1$ and $G_2$ nonempty; equivalently, its complement is
connected. The connected components of $\overline H$ are the \emph{co-components} of
$H$, and $H$ is the join of the subgraphs induced on its co-components.

\begin{lemma}[Clique complexes are closed under joins]
\label{lem:join-clique}
$\Cl(G\vee\Gamma)=\Cl(G)*\Cl(\Gamma)$.
\end{lemma}

\begin{proof}
A vertex set $S=S_G\sqcup S_\Gamma$ of $G\vee\Gamma$ is a clique if and only if $S_G$ is a
clique of $G$ and $S_\Gamma$ is a clique of $\Gamma$, because every cross pair is an edge.
That is precisely the condition for $S_G\cup S_\Gamma$ to be a face of
$\Cl(G)*\Cl(\Gamma)$.
\end{proof}

The next lemma is the single chain-level fact from which everything we do follows. The
chains of a join form the tensor product of the chains of the factors. We fix a vertex
order on $X*Y$ in which every vertex of $X$ precedes every vertex of $Y$.

\begin{lemma}[The augmented chain complex of a join]
\label{lem:join-chains}
Let $\sigma\in\Ct_i(X)$ and $\tau\in\Ct_j(Y)$ be basis simplices, the empty face allowed,
with $i,j\ge-1$. The assignment $\sigma\otimes\tau\mapsto\sigma\cup\tau$ extends to an
isomorphism of graded abelian groups
\begin{equation}
\Ct_k(X*Y;\ZZ)\;\cong\;\bigoplus_{i+j=k-1}\Ct_i(X;\ZZ)\otimes\Ct_j(Y;\ZZ),
\qquad -1\le k\le\dim X+\dim Y+1,
\label{eq:join-chains}
\end{equation}
under which the boundary map of $X*Y$ becomes
\begin{equation}
B(\sigma\otimes\tau)\;=\;B\sigma\otimes\tau\;+\;(-1)^{i+1}\,\sigma\otimes B\tau,
\qquad \sigma\in\Ct_i(X).
\label{eq:join-boundary}
\end{equation}
In other words, $\Ct_*(X*Y;\ZZ)$ is the tensor product of chain complexes
$\Ct_*(X;\ZZ)\otimes\Ct_*(Y;\ZZ)$ with degrees shifted up by one. Over $\RR$ the
isomorphism \eqref{eq:join-chains} is an isometry for the simplex inner products.
\end{lemma}

\begin{proof}
The faces of $X*Y$ are in bijection with pairs $(\sigma,\tau)$, and
$\lvert\sigma\cup\tau\rvert=(i+1)+(j+1)$. This gives \eqref{eq:join-chains} and the
isometry statement. For \eqref{eq:join-boundary}, we write $\sigma=[x_0,\dots,x_i]$ and
$\tau=[y_0,\dots,y_j]$. With the chosen order
$\sigma\cup\tau=[x_0,\dots,x_i,y_0,\dots,y_j]$ and
\[
B(\sigma\cup\tau)=\sum_{a=0}^{i}(-1)^a[\dots\widehat{x_a}\dots,\tau]
+\sum_{b=0}^{j}(-1)^{i+1+b}[\sigma,\dots\widehat{y_b}\dots]
= B\sigma\cup\tau+(-1)^{i+1}\sigma\cup B\tau .
\]
The cases $i=-1$ or $j=-1$ are covered by the augmented convention. For
$\sigma=\emptyset$ the first sum is empty and the sign is $(-1)^0=1$. For a vertex
$\sigma=[x]$ and $\tau=\emptyset$ the formula returns
$B[x]\otimes\emptyset=\emptyset\otimes\emptyset$, the empty face of $X*Y$.
\end{proof}

\subsection{Input and access model}
\label{sec:model}

All problems below receive an unweighted graph on $N$ vertices through the adjacency oracle
\begin{equation}
O_E\ket{u,v,b}\;=\;\ket{u,\,v,\,b\oplus[\{u,v\}\in E]},
\label{eq:edge-oracle}
\end{equation}
where the bracket denotes the edge-membership bit. A classical algorithm may query the
same oracle. This is the succinct model. The graph is a polynomial-size description of a
complex with up to $\binom{N}{k+1}$ simplices in degree $k$. The distinction between
succinct and explicit input is essential. Given the boundary matrices explicitly, integral
homology and its torsion are computable in polynomial time in the matrix size by Smith
normal form, see e.g.\ \cite{Lor08}. All hardness in this area concerns the succinct
model, where the matrices are exponentially large in $N$.

\subsection{The gadget: a projective plane}
\label{sec:gadget}

We need a clique complex with no rational homology and exactly one torsion class. The
smallest triangulation of $\mathbb{RP}^2$ is not a clique complex, but its barycentric
subdivision is. The six-vertex projective plane is the standard torsion-bearing example in
the theory of simplicial spanning trees \cite{Kal83,DKM13}. We recall it in full so that
the constants below are self-contained.

\begin{example}[Six-vertex projective plane]
\label{ex:rp2}
Writing $123$ for $\{1,2,3\}$, let $T$ be the complex on $[6]$ with facets
\[
123,\ 124,\ 135,\ 146,\ 156,\ 236,\ 245,\ 256,\ 345,\ 346 .
\]
Each of the $15$ edges of $K_6$ lies in exactly two facets, every vertex link is a
five-cycle, and $\chi(T)=6-15+10=1$. So $T$ is a closed surface with Euler characteristic
$1$, namely the standard six-vertex triangulation of $\mathbb{RP}^2$. It has
$f(T)=(6,15,10)$, $\Ht_1(T;\ZZ)\cong\ZZ/2$ and $\Ht_q(T;\ZZ)=0$ for $q\neq1$.
\end{example}

$T$ is not itself a clique complex. Its $1$-skeleton is $K_6$, whose clique complex is the
$5$-simplex. Barycentric subdivision repairs this. The subdivision $\sd T$ is the order
complex of the face poset of $T$. It is therefore the clique complex of the comparability
graph $\Gamma$ of that poset, because a chain in a poset is the same thing as a clique in
its comparability graph. Moreover $\lvert\sd T\rvert\cong\lvert T\rvert$
\cite[\S15]{Mun84}. Since $T$ has $6+15+10=31$ nonempty faces, $\sd T$ has $31$ vertices.
Counting chains, it has $90$ edges and $60$ triangles, and $\chi=31-90+60=1$. Throughout,
we set
\[
Y:=\sd T,\qquad \Gamma:=\text{the comparability graph of the face poset of }T,
\qquad Y=\Cl(\Gamma),
\]
and we record
\begin{equation}
\Ht_1(Y;\ZZ)\cong\ZZ/2,\qquad \Ht_q(Y;\ZZ)=0\ (q\neq1),\qquad\text{hence}\qquad
\Ht_*(Y;\QQ)=0 .
\label{eq:gadget}
\end{equation}
We verified all of these statements, and the eigenvalues quoted in \cref{rem:cY}, by
exact computation: Smith normal form of the boundary matrices of $T$ and of $Y$, and exact
diagonalisation of the four augmented Laplacians of $Y$.

We will also need that the gadget cannot be split into a join of smaller pieces, so that
it can in principle be located inside a larger graph.

\begin{lemma}[$\Gamma$ is join-irreducible]
\label{lem:gamma-irreducible}
The complement $\overline\Gamma$ is connected.
\end{lemma}

\begin{proof}
Suppose $\Gamma=A\vee B$ with $A,B$ nonempty. Then every vertex of $A$ is adjacent in
$\Gamma$ to every vertex of $B$, i.e.\ every cross pair is comparable in the face poset of
$T$. Two distinct faces of $T$ of the same dimension are incomparable. So the six vertices
of $T$ lie on one side, the fifteen edges lie on one side, and the ten triangles lie on
one side. A vertex $v$ of $T$ and a triangle not containing $v$ are incomparable, and such
a triangle exists since $v$ lies in only five of the ten triangles. So the vertex level and
the triangle level lie on the same side. Likewise $v$ and an edge not containing $v$ are
incomparable, so the vertex level and the edge level lie on the same side. All three
levels are therefore on one side and the other side is empty, a contradiction.
\end{proof}

\section{Holes   become torsion}
\label{sec:reduction}

We prove two independent results in this section. The first is integral and homological,
the second real and spectral. We derive both from \cref{lem:join-chains} applied to a
join with $Y$.

\subsection{Non-trivial homology classes and torsion}
\label{sec:planting}

\begin{lemma}[The reduction]
\label{lem:planting}
For every finite simplicial complex $X$ and every $n$,
\begin{equation}
\Ht_{n+2}(X*Y;\ZZ)\;\cong\;(\ZZ/2)^{\,\bt_n(X;\FF_2)},
\qquad\text{and}\qquad \Ht_q(X*Y;\QQ)=0\ \text{ for all } q .
\label{eq:planting}
\end{equation}
\end{lemma}

\begin{proof}
By \cref{lem:join-chains}, $\Ht_m(X*Y;\ZZ)$ is the $(m-1)$st homology of the tensor
product of the chain complexes $\Ct_*(X;\ZZ)$ and $\Ct_*(Y;\ZZ)$ of free abelian groups.
The algebraic K\"unneth formula \cite[Theorem~3B.5]{Hat02} therefore gives
\begin{equation}
\Ht_m(X*Y;\ZZ)\;\cong\;\bigoplus_{i+j=m-1}\Ht_i(X;\ZZ)\otimes\Ht_j(Y;\ZZ)\;\oplus\;
\bigoplus_{i+j=m-2}\Tor\bigl(\Ht_i(X;\ZZ),\Ht_j(Y;\ZZ)\bigr).
\label{eq:kunneth-join}
\end{equation}
By \eqref{eq:gadget} only $j=1$ contributes to either sum. Taking $m=n+2$ leaves
\[
\Ht_{n+2}(X*Y;\ZZ)\;\cong\;\bigl(\Ht_n(X;\ZZ)\otimes\ZZ/2\bigr)\;\oplus\;
\Tor\bigl(\Ht_{n-1}(X;\ZZ),\ZZ/2\bigr),
\]
which is the right-hand side of the universal coefficient theorem for $\Ht_n(X;\FF_2)$
\cite[Theorem~3A.3]{Hat02}. Hence $\Ht_{n+2}(X*Y;\ZZ)$ is an $\FF_2$-vector space of
dimension $\bt_n(X;\FF_2)$. For the rational statement, we tensor \eqref{eq:kunneth-join}
with $\QQ$, or apply the K\"unneth formula over the field $\QQ$ directly. Every term is
annihilated, since $\Ht_*(Y;\QQ)=0$.
\end{proof}

\begin{remark}[Other primes]
\label{rem:p-torsion}
Nothing in our proof is special to the prime $2$. Let $Y_p$ be any rationally acyclic
clique complex with $\Ht_1(Y_p;\ZZ)\cong\ZZ/p$ and no other homology. The barycentric
subdivision of any triangulated Moore space $M(\ZZ/p,1)$ will do. The same K\"unneth
computation gives $\Ht_{n+2}(X*Y_p;\ZZ)\cong(\ZZ/p)^{\bt_n(X;\FF_p)}$, so the mod-$p$
Betti number of $X$ becomes $p$-torsion of the join. We restrict to $p=2$ for
concreteness, and because the torsion in the applications of \cref{sec:intro} is
$2$-torsion.
\end{remark}

\begin{remark}[When torsion is easy]
\label{rem:easy}
Combining \cref{lem:planting} with \eqref{eq:uct-count}, the $2$-rank of the torsion
planted in degree $n+2$ is
\[
t_{n+2}(X*Y)\;=\;\bt_n(X;\FF_2)\;=\;\bt_n(X;\QQ)\;+\;t_n(X)\;+\;t_{n-1}(X).
\]
If $X$ is $2$-torsion-free in degrees $n-1$ and $n$, the two correction terms vanish and
$t_{n+2}(X*Y)=\bt_n(X;\QQ)$: on such instances the planted torsion carries exactly a
rational Betti number of the factor $X$, the quantity that Betti number estimation
computes. Our hardness results therefore concern precisely the instances on which this
collapse is not promised.
\end{remark}

\subsection{A constant spectral gap}
\label{sec:gap}

\begin{lemma}[The Laplacian of a join is a Kronecker sum]
\label{lem:join-laplacian}
Under the isometric identification \eqref{eq:join-chains} over $\RR$, the augmented Hodge
Laplacian $\Delta_k(X*Y)$ preserves each summand $\Ct_i(X)\otimes\Ct_j(Y)$, $i+j=k-1$, and
acts on it as
\[
\Delta_i(X)\otimes I\;+\;I\otimes\Delta_j(Y).
\]
Consequently
\begin{equation}
\spec\Delta_k(X*Y)\;=\;\bigcup_{\substack{i+j=k-1\\ \Ct_i(X)\otimes\Ct_j(Y)\neq0}}
\bigl\{\lambda+\mu\;:\;\lambda\in\spec\Delta_i(X),\ \mu\in\spec\Delta_j(Y)\bigr\}.
\label{eq:sumset}
\end{equation}
\end{lemma}

\begin{proof}
Let $\varepsilon$ be the operator acting as $(-1)^{i+1}$ on $\Ct_i(X)$. By
\eqref{eq:join-boundary}, $B=B_X\otimes I+\varepsilon\otimes B_Y$ on the total complex.
Since the identification is an isometry, adjoints may be computed factorwise. Expanding
$\Delta=B^{*}B+BB^{*}$, the cross terms are
\[
(B_X^{*}\varepsilon+\varepsilon B_X^{*})\otimes B_Y\;+\;
(\varepsilon B_X+B_X\varepsilon)\otimes B_Y^{*}\;=\;0,
\]
because $B_X^{*}$ raises and $B_X$ lowers the $X$-degree by one, so each anticommutes with
$\varepsilon$. The remaining terms are
$(B_X^{*}B_X+B_XB_X^{*})\otimes I+\varepsilon^2\otimes(B_Y^{*}B_Y+B_YB_Y^{*})
=\Delta_X\otimes I+I\otimes\Delta_Y$. Restricting to total degree $k-1$ gives the claim.
The two summands of a Kronecker sum commute and are simultaneously diagonalisable in a
product eigenbasis, which gives \eqref{eq:sumset}.
\end{proof}

The same decomposition appears in \cite{KK24} and in \cite[Appendix~E]{BSG+24}. We include
the short proof so that the constant below is fully justified.

\begin{proposition}[Gap from a rationally acyclic join]
\label{prop:gap}
Let $Y$ be a finite complex with $\Ht_*(Y;\RR)=0$ and set
$c_Y:=\min_{-1\le q\le\dim Y}\lambda_{\min}(\Delta_q(Y))$. Then $c_Y>0$, and for
\emph{every} finite complex $X$ and every $k$ with $\Ct_k(X*Y)\neq0$,
\begin{equation}
\lambda_{\min}\bigl(\Delta_k(X*Y)\bigr)\;\ge\;c_Y .
\label{eq:gap}
\end{equation}
Consequently, if the boundary maps of $Z:=X*Y$ are block encoded with normalisation
$\alpha_B$, then for every $r$
\begin{equation}
\zeta^{+}_{\min}\!\left(\frac{B_r(Z)}{\alpha_B}\right)\;\ge\;\frac{\sqrt{c_Y}}{\alpha_B}.
\label{eq:gap-normalized}
\end{equation}
\end{proposition}

\begin{proof}
By the Hodge theorem $\ker\Delta_q(Y)\cong\Ht_q(Y;\RR)=0$, so each augmented
$\Delta_q(Y)$ is positive definite and $c_Y>0$. In \eqref{eq:sumset} every $\lambda\ge0$,
because $\Delta_i(X)$ is positive semidefinite, and every $\mu\ge c_Y$ by definition of
$c_Y$. Hence every eigenvalue of $\Delta_k(X*Y)$ is at least $c_Y$. No hypothesis on $X$
was used. For \eqref{eq:gap-normalized}, the positive squared singular values of $B_r(Z)$
are the nonzero eigenvalues of $L^{\mathrm{up}}_{r-1}(Z)$. These are eigenvalues of
$\Delta_{r-1}(Z)$ (\cref{sec:notation}), hence at least $c_Y$. Dividing by $\alpha_B^2$
gives the claim.
\end{proof}

\begin{remark}[The constant is explicit, absolute and attained]
\label{rem:cY}
For $Y=\sd T$ the four augmented Laplacians have sizes $1,31,90,60$ and smallest
eigenvalues
\[
\lambda_{\min}(\Delta_{-1})=31,\qquad \lambda_{\min}(\Delta_0)=2,\qquad
\lambda_{\min}(\Delta_1)=\lambda_{\min}(\Delta_2)=0.0978869674\ldots
\]
So $c_Y=0.0978869674\ldots$, the smallest root of $x^4-8x^3+19x^2-12x+1$. It is an
eigenvalue of $\Delta_2(Y)=B_2^{*}B_2$ of multiplicity three. The first value is
$B_0B_0^{*}$ on the one-dimensional space $\Ct_{-1}$, namely the number of vertices. The
second is the algebraic connectivity of $\Gamma$, since $\Delta_0=L_\Gamma+J$ with
$L_\Gamma$ the graph Laplacian and $J$ the all-ones matrix. Finally
$\lambda_{\min}(\Delta_1)=\lambda_{\min}(\Delta_2)$, because the nonzero spectrum of
$\Delta_1$ is the union of the nonzero spectra of $B_1^{*}B_1$, which equals that of
$L_\Gamma$ and has least nonzero eigenvalue $2$, and of $B_2B_2^{*}$, which equals that of
$\Delta_2$.

Three points deserve emphasis. First, $c_Y$ is a fixed number, not a function of the
input. \Cref{eq:gap} is a \emph{constant} lower bound on the unnormalised Laplacian of
every instance. The inverse-polynomial factor in \eqref{eq:gap-normalized} enters only
through the block-encoding normalisation $\alpha_B$. We note that this is a different kind
of promise from the one in \cite{KK24}. There the \yes\ instances have a kernel and only
the \no\ instances are gapped. Here \emph{no} instance has a kernel. Second, the bound is
attained, so it cannot be improved in general. For $X=S^0$, two isolated vertices, the
summand $\Ct_0(X)\otimes\Ct_1(Y)$ of $\Ct_2(X*Y)$ carries the eigenvalue $0+c_Y$, since
$\Delta_0(S^0)$ is singular. The same happens for any $X$ with $\Ht_i(X;\RR)\neq0$ in some
degree $i$. Third, the constant is a property of the gadget alone. Any rationally acyclic
clique complex $Y'$ with a single torsion class would serve, with its own $c_{Y'}$. The
mechanism, a fixed factor with positive definite Laplacians lifting the spectrum of a
join, is the one used in \cite[Cor.~1]{BSG+24} and \cite{KK24} to construct
well-conditioned examples. The content of \cref{prop:gap} is that no property of the other
factor is needed.

\end{remark}

\section{Hardness of torsion detection}
\label{sec:hardness}

We reduce from the decision problem shown $\QMAone$-hard in \cite{CK24}. It carries no
spectral promise. The promise in our output instances is supplied entirely by
\cref{prop:gap}.

\begin{problem}[Clique homology, $\DCH_k$ \cite{CK24}]
\label{prob:dch}
\emph{Input:} an unweighted graph $G$ on $N$ vertices via \eqref{eq:edge-oracle} and an
integer $k$. \emph{Decide} whether $\Ht_k(\Cl(G);\QQ)\neq0$. We write $\DCHtwo$ for the
same problem with $\FF_2$ in place of $\QQ$.
\end{problem}

\begin{problem}[Gapped torsion detection, $\GTD_{n,g}$]
\label{prob:gtd}
\emph{Input:} an unweighted graph $H$ on $M$ vertices via \eqref{eq:edge-oracle}, an
integer $n$, and $g>0$. \emph{Promise:} $\Ht_n(\Cl(H);\ZZ)$ is finite, and
$\lambda_{\min}(\Delta_j(\Cl(H)))\ge g$ for every $j$ with $\Ct_j(\Cl(H))\neq0$.
\emph{Decide} whether $\Ht_n(\Cl(H);\ZZ)\neq0$.
\end{problem}

The spectral clause of the promise forces $\Ht_*(\Cl(H);\RR)=0$. On legal instances
\emph{all} integral homology is therefore torsion, and the finiteness clause is implied.
We keep it for readability.

The only obstacle to an unconditional $\QMAone$-hardness proof is that \cite{CK24} is
stated over a field of characteristic zero, while our gadget reads off the mod-$2$ Betti
number. By \eqref{eq:uct-count} the two agree exactly when there is no $2$-torsion in
degrees $k-1$ and $k$. Inspecting the proof of \cref{thm:main} below shows that this
agreement is needed on \no\ instances only: on a \yes\ instance the inequality
$\bt_k(\cdot;\FF_2)\ge\bt_k(\cdot;\QQ)$ of \eqref{eq:uct-count} already suffices.

\begin{problem}[Clique homology with $2$-torsion-free \no\ instances, $\DCHtfno_k$]
\label{prob:dch-tfno}
The promise problem with instances $(G,k)$, $G$ unweighted via \eqref{eq:edge-oracle},
\[
\Pi_{\yes}=\bigl\{(G,k):\Ht_k(\Cl(G);\QQ)\neq0\bigr\},\qquad
\]
\[
\Pi_{\no}=\bigl\{(G,k):\Ht_k(\Cl(G);\QQ)=0\ \text{and}\ t_{k-1}(\Cl(G))=t_k(\Cl(G))=0\bigr\}.
\]
\end{problem}

\begin{conjecture}
\label{conj:torsionfree}
$\DCHtfno_k$ is $\QMAone$-hard.
\end{conjecture}

Shrinking the \no\ set of a promise problem can only make it easier, so
\cref{conj:torsionfree} is implied by the stronger hypothesis that $\DCH_k$ remains
$\QMAone$-hard when \emph{all} instances are restricted to be $2$-torsion-free in degrees
$k-1$ and $k$, and it is implied by $\QMAone$-hardness of $\DCH_k$ on any family of
instances whose \no\ members happen to be $2$-torsion-free in those degrees.
$\QMAone$-hardness of $\DCHtwo$ would serve equally well, and would make
\cref{thm:main}(iii) unconditional. The conjecture is not free. The reduction of
\cite{CK24} runs through the correspondence between clique homology and the zero-energy
ground space of a supersymmetricsystem, that is, through Hodge
theory over $\RR$, and nothing in it controls $\Ht_*(\cdot;\FF_2)$ or the integral torsion
subgroups.

\begin{theorem}
\label{thm:main}
Consider the map $(G,k)\mapsto(G\vee\Gamma,\ n=k+2,\ g=c_Y)$. It adds $31$ vertices and
answers one adjacency query to $G\vee\Gamma$ with one query to $G$ and $O(1)$ reversible
gates. It outputs legal instances of $\GTD_{n,c_Y}$, and:

\begin{enumerate}[label=(\roman*),leftmargin=*,itemsep=1pt]
\item it is a many-one reduction from $\DCHtwo$ to $\GTD_{n,c_Y}$,
\item it is a many-one reduction from $\DCHtfno_k$ to $\GTD_{n,c_Y}$,
\item hence, under \cref{conj:torsionfree}, $\GTD_{n,c_Y}$ is $\QMAone$-hard.
\end{enumerate}
\end{theorem}

\begin{proof}
We write $X:=\Cl(G)$, $W:=\Cl(G\vee\Gamma)$ and $N_W:=N+31$. By \cref{lem:join-clique},
$W=X*Y$. So $W$ is the clique complex of an unweighted graph and the reduction stays in
the input model. One adjacency query to $G\vee\Gamma$ is answered as follows. If both
endpoints lie in $[N]$, we query $G$. If both lie in $V(\Gamma)$, we look up the fixed
graph $\Gamma$. If the endpoints are split, we answer $1$. Since $\Gamma$ is fixed, this
is $O(1)$ reversible gates.

\emph{Validity of the output.} By \cref{lem:planting},
$\Ht_{k+2}(W;\ZZ)\cong(\ZZ/2)^{b}$ with $b:=\bt_k(X;\FF_2)$, so $\Ht_{k+2}(W;\ZZ)$ is
finite, and by \cref{prop:gap} every augmented Laplacian of $W$ has smallest eigenvalue
at least $c_Y$. Hence $(G\vee\Gamma,\,k+2,\,c_Y)$ satisfies the promise of
$\GTD_{n,c_Y}$.

\emph{Correctness of (i).} $\Ht_{k+2}(W;\ZZ)\neq0$ if and only if $b>0$, i.e.\ if and only
if $\Ht_k(X;\FF_2)\neq0$. This is exactly $\DCHtwo$ on $(G,k)$.

\emph{Correctness of (ii).} On a \yes\ instance of $\DCHtfno_k$, $\bt_k(X;\QQ)\ge1$, and
\eqref{eq:uct-count} gives $b=\bt_k(X;\QQ)+t_k(X)+t_{k-1}(X)\ge\bt_k(X;\QQ)\ge1$, since
the two correction terms are nonnegative. So $W$ has torsion. This direction uses no
hypothesis beyond $(G,k)\in\Pi_{\yes}$. On a \no\ instance, $\bt_k(X;\QQ)=0$ and the
\no-side promise gives $t_k(X)=t_{k-1}(X)=0$. Hence $b=0$ and $\Ht_{k+2}(W;\ZZ)=0$.

\emph{(iii)} is (ii) together with \cref{conj:torsionfree}. A polynomial-time many-one
reduction from a $\QMAone$-hard promise problem to $\GTD_{n,c_Y}$ makes the latter
$\QMAone$-hard.
\end{proof}

\begin{corollary}[Unconditional hardness]
\label{cor:np-hard}
$\GTD_{n,c_Y}$ is $\NP$-hard.
\end{corollary}

\begin{proof}
Adamaszek and Stacho \cite{AS12,AS16} prove that it is $\NP$-complete to decide, given a
chordal graph $Q$ and an integer $k$, whether $\Ht_k(\Ind(Q))\neq0$. Here
$\Ind(Q)=\Cl(\overline{Q})$ is the independence complex of $Q$, i.e.\ the clique complex
of its complement. For chordal $Q$ the homology of $\Ind(Q)$ is detected by cross-cycles
\cite{AS12,AS16}: induced subcomplexes isomorphic to the boundary of a cross-polytope
that contain a maximal face of $\Ind(Q)$. In the language of $Q$ these are the
\emph{strong} induced matchings, the induced matchings whose vertex set contains a
maximal independent set of $Q$. The maximality clause is what carries the hardness: a
maximum induced matching of a chordal graph can be found in polynomial time
\cite{Cam89}, whereas deciding whether a chordal graph has a strong induced matching of a
prescribed size is $\NP$-complete \cite{AS16}. Independence complexes of chordal graphs
are vertex-decomposable \cite{Woo09}, hence homotopy equivalent to wedges of spheres or
contractible \cite{Kaw10,Ada17}. Their integral homology is therefore free, and
$\Ht_k(\Ind(Q);\FF_2)\neq0$ if and only if $\Ht_k(\Ind(Q);\QQ)\neq0$. Hence $\DCHtwo$ is
$\NP$-hard on the instances $(\overline Q,k)$. One adjacency query to $\overline Q$ is one
query to $Q$ followed by a NOT gate, and \cref{thm:main}(i) transfers the hardness to
$\GTD_{n,c_Y}$.
\end{proof}

\begin{remark}[What is conditional and what is not]
\label{rem:conditional}
The \yes\ direction of \cref{thm:main}(ii) needs only $\bt_k(\cdot;\FF_2)\ge\bt_k(\cdot;\QQ)$,
immediate from \eqref{eq:uct-count} and valid for every complex. The \no\ direction is
where $2$-torsion-freeness comes in, and it comes in only there. The reason is that
$\bt_k(\cdot;\QQ)=0$ does not force $\Ht_k(\cdot;\FF_2)=0$: $2$-torsion in either of two
adjacent degrees of the input creates a spurious mod-$2$ class, hence spurious torsion
downstream. The unconditional content of our theorem is therefore that \emph{gapped
torsion detection is at least as hard as mod-$2$ clique homology}. $\QMAone$-hardness
follows if that problem is $\QMAone$-hard, or if the \no\ instances of a $\QMAone$-hard
family of rational clique-homology instances can be certified $2$-torsion-free.
\end{remark}

\section{Outlook and open problems}
\label{sec:open}

Torsion is the part of homology that quantum (and classical) TDA has disregarded. We have shown that deciding whether a clique
complex has torsion is as hard as mod 2 Betti number computation. We close with the questions we consider most pressing.

\begin{enumerate}[leftmargin=*,itemsep=3pt]
\item\label{op:containment} \emph{Containment.} Is $\GTD_{n,g}\in\QMA$? For gapped
rational clique homology on weighted graphs, containment holds \cite{KK24} via the
Laplacian ground-energy formulation, which is unavailable for torsion. Our instances are
gapped, so the usual obstruction to containment is absent. What is missing is a verifier.
A containment result would complete the classification and would not conflict with
\cref{thm:main}.

\item\label{op:conjecture} \emph{Discharging \cref{conj:torsionfree}.} Either show that
the \no\ instances of \cite{CK24}, or a modification of them, are $2$-torsion-free in the
two relevant degrees, or prove $\QMAone$-hardness of mod-$2$ clique homology directly.
Either route makes \cref{thm:main}(iii) unconditional.
\item\label{op:unweighted} \emph{The unweighted gapped rational problem.} \cite{KK24}
leave open whether gapped clique homology is $\QMAone$-hard for unweighted graphs.
\Cref{prop:gap} shows that joining with a rationally acyclic factor imposes a constant gap
on an unweighted complex at no cost. It also kills all rational homology. Is there a
gadget that imposes a gap while preserving rational homology in one degree?
\end{enumerate}


\section{Acknowledgements}
We wish to acknowledge the helpful discussions with Dimitrios Thanos and Caesnan Leditto, who suggested the gadget might be useful for investigating torsion.

\end{document}